\documentclass[10pt, onecolumn]{IEEEtran}

\usepackage{amsmath}
\usepackage{amssymb}
\usepackage{amsfonts}
\usepackage{epsfig}
\usepackage{subfigure}
\usepackage{calc}
\usepackage{color}
\usepackage[all]{xy}
\usepackage{bm}
\usepackage{algorithmicx}
\usepackage{algpseudocode}
\usepackage{algorithm}
\usepackage{enumerate}
\usepackage{pdfsync}
\usepackage{framed}
\usepackage{caption}
\usepackage{amsthm}
\usepackage{blkarray}
\usepackage{float}
\newtheorem{defn}{Definition}

\newtheorem{thm}{Theorem}[section]
\newtheorem{cor}[thm]{Corollary}
\newtheorem{prop}{Proposition}

\newtheorem{lem}[thm]{Lemma}
\newtheorem{conj}[thm]{Conjecture}
\newtheorem{constr}[thm]{Construction}
\newtheorem{note}{Remark}
\newtheorem{example}{Example}
\newcommand{\bit}{\begin{itemize}}
\newcommand{\eit}{\end{itemize}}
\newcommand{\bcor}{\begin{cor}}
\newcommand{\ecor}{\end{cor}}
\newcommand{\beq}{\begin{equation}}
\newcommand{\eeq}{\end{equation}}
\newcommand{\beqn}{\begin{equation*}}
\newcommand{\eeqn}{\end{equation*}}
\newcommand{\bea}{\begin{eqnarray}}
\newcommand{\eea}{\end{eqnarray}}
\newcommand{\bean}{\begin{eqnarray*}}
\newcommand{\eean}{\end{eqnarray*}}
\newcommand{\ben}{\begin{enumerate}}
\newcommand{\een}{\end{enumerate}}
\newcommand{\bdefn}{\begin{defn}}
\newcommand{\edefn}{\end{defn}}
\newcommand{\bnote}{\begin{note}}
\newcommand{\enote}{\end{note}}
\newcommand{\bprop}{\begin{prop}}
\newcommand{\eprop}{\end{prop}}
\newcommand{\blem}{\begin{lem}}
\newcommand{\elem}{\end{lem}}
\newcommand{\bthm}{\begin{thm}}
\newcommand{\ethm}{\end{thm}}
\newcommand{\bconj}{\begin{conj}}
\newcommand{\econj}{\end{conj}}
\newcommand{\bconstr}{\begin{constr}}
\newcommand{\econstr}{\end{constr}}
\newcommand{\bpf}{\begin{proof}}
\newcommand{\epf}{\end{proof}}

\title{On MBR codes with replication}

\begin{document}
	\title{On MBR codes with replication}
	
	\author{
		\IEEEauthorblockN{M. Nikhil Krishnan and P. Vijay Kumar, \it{Fellow}, \it{IEEE}}
		
		\IEEEauthorblockA{Department of Electrical Communication Engineering, Indian Institute of Science, Bangalore.  \\ Email: nikhilkrishnan.m@gmail.com, pvk1729@gmail.com} 
		\thanks{P. Vijay Kumar is also an Adjunct Research Professor at the University of Southern California.  This work is supported in part by the National Science Foundation under Grant No. 1421848 and in part by the joint UGC-ISF research program.}
	}
	\maketitle

\begin{abstract}
	An early paper by Rashmi et. al. presented the construction of an $(n,k,d=n-1)$ MBR regenerating code featuring the inherent double replication of all code symbols and repair-by-transfer (RBT), both of which are important in practice.  We first show that no MBR code can contain even a single code symbol that is replicated more than twice.  We then go on to present two new families of MBR codes which feature double replication of all systematic message symbols.  The codes also possess a set of $d$ nodes whose contents include the message symbols and which can be repaired through help-by-transfer (HBT).  As a corollary, we obtain systematic RBT codes for the case $d=(n-1)$ that possess inherent double replication of all code symbols and having a field size of $O(n)$ in comparison with the general, $O(n^2)$ field size requirement of the earlier construction by Rashmi et. al.  For the cases $(k=d=n-2)$ or $(k+1=d=n-2)$, the field size can be reduced to $q=2$, and hence the codes can be binary.  We also give a necessary and sufficient condition for the existence of MBR codes having double replication of all code symbols and also suggest techniques which will enable an arbitrary MBR code to be converted to one with double replication of all code symbols. 
\end{abstract}

\begin{keywords}
	Distributed storage, regenerating codes, exact repair, MBR codes.
\end{keywords}

\section{Introduction}\label{sec:intro}

Dimakis et. al. \cite{DimGodWuWaiRam} introduced a class of distributed storage codes called \textit{regenerating codes}. In the regenerating code framework, a file of $B$ symbols will be encoded to a vector-code having $n$ nodes storing $\alpha$ symbols each. During a single node failure, the failed node can be regenerated by downloading $\beta\leq \alpha$ symbols each, from \textit{any} $d$ surviving nodes (\textit{node-repair property}). Also, by accessing symbols in \textit{any} $k$ nodes, the file can be retrieved (\textit{data-collection property}). \cite{DimGodWuWaiRam} showed existence of a trade-off between $\alpha$ (storage) and $d\beta$ (bandwidth) for given $n$, $k$, $d$, $\beta$ and file-size $B$ given by
\begin{equation}
B \leq \sum_{i=0}^{k-1}{min\{\alpha,(d-i)\beta\}} \label{filesizebound} .
\end{equation}
The \textit{Minimum Storage Regeneration (MSR)} and \textit{Minimum Bandwidth Regeneration (MBR)} points are the two extremal points in the trade-off, where $\alpha$ and $d\beta$ are minimized first respectively.

There are two models of repair for regenerating codes; \textit{functional} and \textit{exact repair}. The bound \eqref{filesizebound} in \cite{DimGodWuWaiRam} was based on the functional repair model. In this model, the contents of a node are permitted to change following repair, while retaining the node-repair and data-collection properties. In the exact repair case, node contents remain the same after repair.  

 At the MBR point, $\alpha$, $\beta$ are given respectively by: $\alpha =  d\beta$, $\beta = \frac{B}{kd-{k\choose 2}}$.
Our focus here is on MBR codes and exact repair, with parameter $\beta=1$.

During node-repair, if repair is carried out simply by {\em reading \em} precisely $1$ symbol from each of the $d$ nodes, the repair is termed as \textit{help-by-transfer (HBT)}.  If in addition to help-by-transfer, no computation is needed at the replacement node either, we will speak of \textit{repair-by-transfer (RBT)}\footnote{some authors use RBT notation for both the scenarios}. Minimizing reads during repair is advantageous, as it translates to savings in the utilization of computational resources, storage-disk durability etc.

Replication of data is a second important consideration in practical distributed storage platforms such as Hadoop \cite{hadoop}. In addition to improving resiliency against errors and erasures, it also helps with data availability when there are transient failures in the system. It also increases chances of local data computation in nodes (data locality), thereby reducing job execution delays (jobs are compute operations to be performed on the data) and bandwidth \cite{mapreduce}.

In \cite{ElrRam}, fractional repetition (FR) codes are introduced which generalize the RBT-MBR construction in \cite{RasShaKumRam_allerton}, focusing more on the replication aspect and relaxing the requirement of {\em any $\mathit{d}$} for repair to a table-based repair-model involving $d$ nodes. In contrast to FR codes, our constructions provide a generalization while staying within the MBR regime, and retaining properties such as reduced reads, systematicity of codes and the best-possible MBR replication level of $2$.

\subsection{Other Related Work}
The Product-Matrix MBR (PM-MBR) codes \cite{RasShaKum_pm} are a family of MBR codes that exists for all possible parameter sets $(n,k\leq d,d\leq n-1,\beta=1)$.  $(n,k\leq d,d=n-1$)-MBR codes having the RBT property for all node-repairs are provided in \cite{RasShaKumRam_allerton} and \cite{ShaRasKumRam_rbt}. These codes are formed via concatenation of an outer MDS code of length ${n \choose 2}$ and an inner replication-$2$ code. In \cite{NovelRBRMBRLinChung}, the authors present repair-by-transfer MBR codes for $d=(n-1)$ based on congruent transformations on skew-symmetric matrices. This approach requires only $O(n)$ field-size and lesser computational complexity over the RBT construction in \cite{ShaRasKumRam_rbt}. \cite{UpdateEfficientPM} gives new encoding matrices for MSR and MBR codes that have the least update complexity in PM framework. \cite{UpdateEfficientPM} also proposes new decoding schemes that have improved error correction capability.  The paper \cite{NiharHBTnonexistence} proves non-existence of $d<(n-1)$ MBR codes  with HBT for all nodes. This paper also gives PM based constructions for two relaxations, namely, HBT for only a specific set of $d$ nodes and HBT recovery from $d$ specific nodes (for all nodes). 

\subsection{Our Contributions}

We present here two new families of $(n,k\leq d,d\leq n-2,\beta=1)$ MBR constructions.  Both constructions have an RBT-MBR code \cite{RasShaKumRam_allerton} contained in $(d+1)$ nodes, as a component of the construction. These $(d+1)$ nodes contain systematic data with double replication and a subset $d$ of these nodes can be repaired via HBT, irrespective of the choice of $d$ helpers.

In Section-\ref{nonexistenceRep}, we show that it is not possible to have a replication level $>2$ even for a symbol in an MBR code, when $k\geq2$.  Section-\ref{sec:constructionGeneral} describes the two MBR constructions. The first family of codes are motivated by the RBT construction appearing in \cite{RasShaKumRam_allerton} and \cite{ShaRasKumRam_rbt}. For $(n,k=n-2,d=n-2)$ and $(n,k=n-3,d=n-2)$, these codes can be implemented over $\mathbb{F}_2$. The second family of codes are based on internal node transformations of PM-MBR, combining ideas from \cite{ShaRasKumRam_rbt} and \cite{NiharHBTnonexistence}. This construction also has the replication, HBT properties. Systematicity of the code is ensured by an additional precoder, which is related to the underlying MDS code appearing in the PM-MBR generator matrix. For $d=k$, both codes have smaller update complexity than any code in the conventional PM framework. In Section-\ref{sec:MBRwithDoubleRep}, we discuss, to what extent double replication can be brought in an MBR code with given parameters $(n,k,d,\beta=1)$ and modified constructions to achieve that.

\section{Non-existence of MBR codes with Replication $>$ 2}\label{nonexistenceRep}

For $i\in[n]\triangleq\{1,2,\ldots,n\}$, let $W_i$, ${}_{\mathcal{D}}S_i^j$ denote random variables corresponding to node-$i$ content and repair-data supplied by node-$i$ ($i\subseteq\mathcal{D}$, $|\mathcal{D}|=d$) to repair node-$j$, respectively. ${}_{\mathcal{D}}S_i^j$ is a function of $W_i$. Let $k\geq2$. From \cite{ShaRasKumRam_rbt}, we list a few equations on entropies of these random variables as the following lemma.  
\begin{lem}\label{lem:entropylemmas}
	\begin{enumerate}
	\item	$H(W_i) =   \alpha = d\beta$
	\item $H(W_i|W_j) = H(W_i|{}_{\mathcal{D}}S_j^i) =(d-1)\beta$
	\item  $H({}_{\mathcal{D}}S_j^i,{}_{\mathcal{D}}S_k^i)   =   2\beta$	
\end{enumerate}
\end{lem}
\begin{thm}
	It is not possible to replicate a symbol more than twice in an exact-MBR code with $k\geq2$. 
\end{thm}

\begin{proof}
	
	WOLOG assume a symbol $x$ is present in nodes-$1$, $2$ and $3$. Note that Lemma-\ref{lem:entropylemmas}.$1$ implies:
	\begin{equation*}
	H(x)=1
	\end{equation*} 
	
	From Lemma-\ref{lem:entropylemmas}.$2$, 

	\begin{eqnarray*}
	(d-1)\beta & = & H(W_1|W_2)\\
	& = & H(W_1|W_2,x)\\
	& \le & H(W_1|{}_{\mathcal{D}}S_2^1,x)\\
	& \le &  H(W_1|{}_{\mathcal{D}}S_2^1)\\
	& = & (d-1)\beta.
\end{eqnarray*}
	
	 $\therefore H(W_1|{}_{\mathcal{D}}S_2^1) - H(W_1|{}_{\mathcal{D}}S_2^1,x) =0.$

	Hence, 
	\begin{eqnarray*}
	H(x|{}_{\mathcal{D}}S_2^1)&= & H(x|{}_{\mathcal{D}}S_2^1,W_1)+I(x;W_1|{}_{\mathcal{D}}S_2^1)\\
	 & = & H(x|{}_{\mathcal{D}}S_2^1,W_1)+H(W_1|{}_{\mathcal{D}}S_2^1)-H(W_1|{}_{\mathcal{D}}S_2^1,x)\\
	 & = & H(x|{}_{\mathcal{D}}S_2^1, W_1)\\
	 & \stackrel{(a)}{=} & 0 
    \end{eqnarray*}
    where $(a)$ follows as node-$1$ contains $x$. Similarly, $H(x|{}_{\mathcal{D}}S_3^1)$  $=$ $0$. 
	
	Therefore, we have
	\begin{eqnarray*}
		H({}_{\mathcal{D}}S_2^1,{}_{\mathcal{D}}S_3^1) & = & H({}_{\mathcal{D}}S_2^1)+H({}_{\mathcal{D}}S_3^1|{}_{\mathcal{D}}S_2^1)\\
		& = & H({}_{\mathcal{D}}S_2^1)+H({}_{\mathcal{D}}S_3^1|{}_{\mathcal{D}}S_2^1,x)+I({}_{\mathcal{D}}S_3^1;x|{}_{\mathcal{D}}S_2^1)\\
		& = & H({}_{\mathcal{D}}S_2^1)+H({}_{\mathcal{D}}S_3^1|{}_{\mathcal{D}}S_2^1,x)+H(x|{}_{\mathcal{D}}S_2^1)-H(x|{}_{\mathcal{D}}S_2^1,{}_{\mathcal{D}}S_3^1)\\
		& = & H({}_{\mathcal{D}}S_2^1)+H({}_{\mathcal{D}}S_3^1|{}_{\mathcal{D}}S_2^1,x)\\
		& \leq & H({}_{\mathcal{D}}S_2^1)+H({}_{\mathcal{D}}S_3^1|x)\\
		& = & H({}_{\mathcal{D}}S_2^1)+H({}_{\mathcal{D}}S_3^1)-\big(H(x)-H(x|{}_{\mathcal{D}}S_3^1)\big)\\
		& = & \beta+\beta-(1-0)\\
		& = & 2\beta-1.
	\end{eqnarray*}
	This contradicts Lemma-\ref{lem:entropylemmas}.$3$ and proves that it is not possible to replicate a symbol  more than twice.	
\end{proof}
\begin{note}
	For $k=1$, $\exists$ MBR codes with replication $>2$. For eg., consider a (vector) replication code with $d$ distinct symbols in a node, and every node $1,2,\ldots,n$ identical. This is clearly an $(n,k=1,d,\beta=1)$ MBR code with each symbol having replication $n$.
\end{note}

\section{Product-Matrix MBR and Repair-By-Transfer MBR codes}\label{recapSection}
In this section, we give a brief summary of PM-MBR codes \cite{RasShaKum_pm} and RBT-MBR codes appearing in \cite{RasShaKumRam_allerton}, \cite{ShaRasKumRam_rbt}.

\subsection{Product-Matrix (PM) MBR \cite{RasShaKum_pm}}\label{PMdiscussion}
Consider a symmetric $(d\times d)$ message matrix $M$. 
\begin{align}
\boldsymbol{M} &= 
\begin{bmatrix}
\bf{S} & \bf{T}\\           
$\bf{T}$^t & \bf{0} 
\end{bmatrix}\label{eq:PMmessagematrix}
\end{align} 
where $\bf{S}$ is a $(k\times k)$ symmetric matrix holding ${k+1\choose2}$ message symbols, $\bf{T}$ is a $(k\times(d-k))$ matrix holding $k(d-k)$ message symbols. There is a $(d\times n)$ encoding matrix $\boldsymbol{\psi}$$=$$\boldsymbol{[\phi_{PM}\ \Delta_{PM}]$}$^t$, where $\boldsymbol{\phi_{PM}}$ and $\boldsymbol{\Delta_{PM}}$ are $(n\times k)$ and $(n\times (d-k))$ matrices respectively. Also, any $k$ rows of $\boldsymbol{\phi_{PM}}$ and any $d$ columns of $\boldsymbol{\psi}$ are linearly independent. Let ${\boldsymbol{\psi}}_i$ denote $i^{th}$ column of $\boldsymbol{\psi}$. Node-$i$ content is given by:
\begin{align}
	{\bf{n}}_i={\bf{M}}{\boldsymbol{\psi}}_i\label{eq:PM_nodecontent}
\end{align}
For the systematic version of Product-Matrix MBR,
\begin{align}
\boldsymbol{\psi} &= 
\begin{bmatrix}
\bf{I_k} & \bf{0} \\
{\boldsymbol{\tilde{\phi}}} & \boldsymbol{\tilde{\Delta}} \\
\end{bmatrix}^t\label{eq:PMsystematicMatrix}
\end{align} 
where $\bf{I_k}$ is the $(k\times k)$ identity matrix, $\bf{0}$ is a $(k\times (d-k))$ matrix and $\boldsymbol{[\tilde{\phi}\  \ \tilde{\Delta}]}$ is an $((n-k)\times d)$ Cauchy matrix. The repair-data transmitted from node-$i$ to $j$ (or vice-versa) during repair is given by:
\begin{align}
S_i^j &= {\boldsymbol{\psi}}_j^t{\bf{M}}{\boldsymbol{\psi}}_i\label{eq:PM_repairdata}
\end{align}
\subsection{Repair-by-transfer MBR codes for $d=(n-1)$ \cite{ShaRasKumRam_rbt}}\label{RBTdiscussion}
The construction can be visualized in terms of a complete graph on $n$ nodes. $B$ message symbols will be encoded using an [${n\choose 2}, B$]-MDS code. A code symbol will be assigned to each edge of the graph and a node will be storing symbols assigned to the edges incident on it. During single-node failures, repair can be performed by merely transferring one symbol each, from the remaining $(n-1)$ nodes.
\section{MBR constructions with double replication for systematic message symbols}\label{sec:constructionGeneral}

\subsection{Construction-A for $k=d\leq(n-2)$}\label{deqkltnm2}
Consider a field $\mathbb{F}_q$ with characteristic $2$. For $d=k$, $B = kd - {k \choose 2}={d+1 \choose 2}$. The starting point of this construction will be the code described in \ref{RBTdiscussion}. The code in \ref{RBTdiscussion} can be viewed in terms of a $(d\times d)$ symmetric matrix $\bf{M}$ and it's diagonal. ${d+1\choose 2}$ message symbols $\{m_i\}_{i=1}^{B}$ will be used to populate $\bf{M}$. Column entries of the augmented matrix $\bf{[M|\text{diag}(M)]}$ gives an equivalent description of the code in \ref{RBTdiscussion}, for parameters $(n_{RBT}=d+1,k_{RBT}=d,d_{RBT}=d)$. These $(d+1)$ columns (having systematic message symbols with replication-$2$) form the first $(d+1)$ node contents ${\bf{n}}_1$, ${\bf{n}}_2$, $\ldots$, ${\bf{n}}_{(d+1)}$. 

Consider a $(d\times(n-(d+1)))$ Cauchy matrix $\boldsymbol{\phi}$ over $\mathbb{F}_q$ (char. 2). ${\boldsymbol{\phi}}=[{\boldsymbol{\phi}}_{(d+2)}|{\boldsymbol{\phi}}_{(d+3)}|\ldots|{\boldsymbol{\phi}}_n]$ (where ${\boldsymbol{\phi}}_i$$ \in \mathbb{F}_q^{d\times1}$). Let $[{\boldsymbol{\phi}}_{1}|{\boldsymbol{\phi}}_{2}|\ldots|{\boldsymbol{\phi}}_d]=\bf{I_d}$. Node-$i$ $\in [n]\setminus\{(d+1)\}$ will store ${\bf{n}}_i={\bf{M}}{\boldsymbol{\phi}}_i=[{\bf{n}}_1|{\bf{n}}_2|\ldots|{\bf{n}}_d]{\boldsymbol{\phi}}_i$.
\begin{equation}
[{\bf{n}}_1|{\bf{n}}_2|\ldots|{\bf{n}}_n]=[{\bf{M}}|\text{diag}{\bf{(M)}}| {\bf{M}}{\boldsymbol{\phi}}]\nonumber
\end{equation} 

\begin{note}\label{rem:PMplusDiag}
	As $d=k$, it is clear from \eqref{eq:PMsystematicMatrix} that $[{\bf{n}}_1|{\bf{n}}_2|\ldots|{\bf{n}}_d|{\bf{n}}_{(d+2)}|\ldots|{\bf{n}}_n]$ gives a PM-MBR code on $(n-1)$ nodes, with ${\boldsymbol{\psi}}_j={\boldsymbol{\phi}}_j$ $1\leq j\leq n$, $j\neq (d+1)$. Also, any $d$-subset of $\{{\boldsymbol{\phi}}_j\}$ is independent. 
\end{note}
\begin{lem}
	${({\boldsymbol{\phi}}_i\odot{\boldsymbol{\phi}}_i)}^t{\bf{n}}_{(d+1)}={\boldsymbol{\phi}}_i^t{\bf{n}}_i$, $1\leq i\leq n$, $i\neq (d+1)$ ($\odot$ indicates element-wise multiplication).\label{diagonalNodeHelperDataLemma}
\end{lem}
\begin{proof}
	RHS =  	${\boldsymbol{\phi}}_i^t{\bf{M}}{\boldsymbol{\phi}}_i$ = $\sum_{a=1}^{d}{\sum_{b\neq a}{[{\boldsymbol{\phi}}_i]_aM_{a,b}[{\boldsymbol{\phi}}_i]_b}}+\sum_{a=1}^{d}{[{\boldsymbol{\phi}}_i]_aM_{a,a}[{\boldsymbol{\phi}}_i]_a}$$\stackrel{(b)}{=}$$\sum_{a=1}^{d}{[{\boldsymbol{\phi}}_i]_aM_{a,a}[{\boldsymbol{\phi}}_i]_a}$ = LHS,
    where $(b)$ follows as $\bf{M}$ is symmetric and $\mathbb{F}_q$ has char. $2$.
	
\end{proof}

\begin{prop}
	The vector-code construction described above is MBR with parameters $(n, k, d=k, \beta = 1, \alpha = d)$
\end{prop}

\begin{proof}
{\em Node-repair\em}: From Remark-\ref{rem:PMplusDiag}, it is clear that only repairs involving ${\bf{n}}_{(d+1)}$ need to be checked. 

Assume node-$(d+1)$ failed. Let an arbitrary set $\mathcal{D}=\{i_1,i_2,\ldots,i_d\}$ of $d$ nodes, ($i_j\neq (d+1)$), help in the repair. Each node-$i$ $\in$ $\mathcal{D}$ provides ${\boldsymbol{\phi}}_i^t{\bf{n}}_i$. Applying Lemma-\ref{diagonalNodeHelperDataLemma}, replacement node has access to ${\boldsymbol{\phi}}_{rep}^t{\bf{n}}_{(d+1)}$, where ${\boldsymbol{\phi}}_{rep}={{\boldsymbol{\phi}}}_\mathcal{D}\odot {{\boldsymbol{\phi}}}_\mathcal{D}$, ${\boldsymbol{\phi}}_{\mathcal{D}}=[{\boldsymbol{\phi}}_{i_1}\ {\boldsymbol{\phi}}_{i_2}\ \ldots {\boldsymbol{\phi}}_{i_d}]$. Invertibility of ${\boldsymbol{\phi}}_{rep}$ follows from the invertibility of  ${\boldsymbol{\phi}}_\mathcal{D}$, as field characteristic is $2$. Thus, replacement node can retrieve ${\bf{n}}_{(d+1)}$ from ${\boldsymbol{\phi}}_{rep}^t{\bf{n}}_{(d+1)}$.

Now, consider the failure of node-$i$ ($i\neq (d+1)$). Let the helper nodes come from a set $\mathcal{D'}=\{i_1,i_2,\ldots,i_{(d-1)}\}\cup\{(d+1)\}$ of $d$ nodes, ($i_j\neq i$). Helper node-$(d+1)$ provides ${({\boldsymbol{\phi}}_i\odot{\boldsymbol{\phi}}_i)}^t{\bf{n}}_{(d+1)}$. Any other helper node-$i_j$ gives ${\boldsymbol{\phi}}_{i}^t{\bf{n}}_{i_j}={\boldsymbol{\phi}}_{i}^t{\bf{M}}{\boldsymbol{\phi}}_{i_j}={\boldsymbol{\phi}}_{i_j}^t{\bf{M}}{\boldsymbol{\phi}}_{i}={\boldsymbol{\phi}}_{i_j}^t{\bf{n}}_i$. Applying Lemma-\ref{diagonalNodeHelperDataLemma}, replacement node has access to ${\boldsymbol{\phi}}_{\mathcal{D}'}^t{\bf{n}}_i$, where ${\boldsymbol{\phi}}_{\mathcal{D}'}=[{\boldsymbol{\phi}}_i\ {\boldsymbol{\phi}}_{i_1}\ {\boldsymbol{\phi}}_{i_2}\ \ldots {\boldsymbol{\phi}}_{i_{(d-1)}}]$. Clearly ${\boldsymbol{\phi}}_{\mathcal{D}'}$ is invertible and hence, replacement node can retrieve ${\bf{n}}_i$ from ${\boldsymbol{\phi}}_{\mathcal{D}'}^t{\bf{n}}_i$.

{\em Data-collection\em}: As $d=k$, this follows from node-repair property.

\end{proof}
\begin{note}
	\emph{(Help-by-transfer)} To repair a node-$i$ in $[d]$, a helper node-$j$ will be providing ${\boldsymbol{\phi}}_{i}^t{\bf{n}}_{j}$ or ${(\boldsymbol{\phi}}_{i}\odot {\boldsymbol{\phi}}_{i})^t{\bf{n}}_{j}$. As  ${\boldsymbol{\phi}}_{i}={\bf{e}}_i$, $i^{th}$ column of $\bf{I_d}$, this is equivalent to providing a stored symbol with out computation.
\end{note}
\begin{note}\label{rem:binarykeqd}
	When $d=k=(n-2)$, $\boldsymbol{\phi}={\boldsymbol{\phi}}_n$ can be chosen to be $\boldsymbol{1}$ (all-one vector) and hence, implementation over $\mathbb{F}_2$ is possible.
\end{note}

\subsection{Linearized Polynomials \cite{LidlFiniteFields}}

A polynomial of the form $f(x)=\sum_{i=0}^{t} a_ix^{q^i}$, $a_{t}\neq 0$ is called a linearized polynomial of $q$-degree $t$. The coefficients \{$a_i$\} are coming from $\mathbb{F}_{q^m}$. It is known that $f(b_1x_1+b_2x_2)=b_1f(x_1)+b_2f(x_2)$, where $b_i\in \mathbb{F}_q$, $x_i\in \mathbb{F}_{q^m}$. A linearized polynomial with $q$-degree $t$ can be uniquely determined from $(t+1)$ evaluations at points $\{\theta_i\}_{i=1}^{t+1} \subseteq \mathbb{F}_{q^m}$, which are independent over $\mathbb{F}_{q}$. \cite{Gab} gives maximum rank distance Gabidulin codes based on linearized polynomial evaluations.

\subsection{Construction-A for $k<d\leq (n-2)$}\label{kltdltnm2}
The file-size, $B=kd - {k \choose 2}$ symbols in $\mathbb{F}_{q^m}$ (with char. $2$). Here, $q^m$ is chosen in such a way that $\exists$ $n_c=d^2 - {d \choose 2} = {d+1 \choose 2}$ symbols $\{\theta_i\}_{i=1}^{n_c}\subseteq \mathbb{F}_{q^m}$, which are independent over $\mathbb{F}_{q}$. A systematic [$n_c$, $B$, $n_c-B+1$]-Gabidulin code will be used to encode $B$ message symbols to $n_c$ code symbols $\{c_i\}_{i=1}^{n_c}$. A vector-code $\mathcal{C}$ will be constructed treating these code symbols as message symbols $\{m_i\}$ in \ref{deqkltnm2} (with $\boldsymbol{\phi}$ over $\mathbb{F}_q$).

\begin{prop}
	$\mathcal{C}$ is an MBR code with parameters $(n, k<d, d\leq(n-2), \beta = 1, \alpha = d)$.
\end{prop}
\begin{proof}
	
	{\em Node-repair\em}: Node-repair property is inherited from the underlying construction in \ref{deqkltnm2}.
	
	{\em Data-collection (DC)\em}: 
	If all the $k$ nodes are in $[d+1]$, DC follows from the embedded RBT code. Consider the case when $k'<k$ nodes are read from $[d+1]$ and $(k-k')$ nodes are read from $\mathcal{K}''$$\subseteq$ $[n]\setminus[d+1]$ ($|\mathcal{K}''|=k-k'$). Using the notation in \ref{deqkltnm2}, each node-$j$ in $\mathcal{K}''$ has access to symbols $\{{\boldsymbol{\phi}}_j^t{\bf{n}}_i\}_{i\in[d]} \cup\{({\boldsymbol{\phi}}_j\odot{\boldsymbol{\phi}}_j)^t{\bf{n}}_{(d+1)}\}$, which are $(d+1)$ equations on $n_c$ variables $\{c_i\}_{i=1}^{n_c}$. Set of variables $\mathcal{V}_i\subseteq\{c_i\}_{i=1}^{n_c}$ arising in each equation-$i$ can be determined using a complete graph $\mathcal{G}$ on $(d+1)$ vertices. $\mathcal{G}$ will be identical to the `RBT graph' (described in \ref{RBTdiscussion}) contained in the first $(d+1)$ nodes of $\mathcal{C}$ (Fig. \ref{fig:gen_example_n_eq_8} gives an example). After removing the (known) interference from $k'$ nodes, let $\mathcal{V}_i'$ denote the variable set corresponding to a (modified) equation-$i$ coming from a node in $\mathcal{K}''$. $\{\mathcal{V}_i'\}$ are  determined by the induced (complete) graph $\mathcal{G}'$ of $\mathcal{G}$ on the remaining vertices in $[d+1]$ (illustrated in Fig. \ref{fig:gen_example_n_eq_8}). Contents of $\mathcal{V}_i'$ are exactly the edges incident on a vertex-$i$ in $\mathcal{G}'$. Let $i_1, i_2, \ldots, i_{(d-k'+1)}$ be the vertices in $\mathcal{G}'$. Any one among the $(d-k'+1)$ equations given by a node in $\mathcal{K}''$ is dependent on the others (application of Lemma-\ref{diagonalNodeHelperDataLemma}). Hence, one can remove a vertex (i.e., a modified equation) $i_{(d-k'+1)}$ from $\mathcal{G}'$. All the remaining equations are independent as there is an unshared edge (unshared $c_j$) for each vertex. Each vertex contains $(d-k')$ $c_i$'s, i.e., $|\mathcal{V}_i'|=d-k'$. 
	
	Now, we will obtain a lower bound on the number of independent equations arising from all the $(k-k')$ nodes $\in$ $\mathcal{K}''$. Let equations corresponding to $\{i_1,\ldots,i_{l-1}\}$ be already added to a set $\mathcal{S}$. Considering equations corresponding to $i_l$ ($1\leq l\leq(d-k')$) from all nodes $\in$ $\mathcal{K''}$, there will be {\em at least\em} $min\{(d-k'-(l-1)), (k-k')\}$ independent equations that can be added to $\mathcal{S}$, as $|\mathcal{V}_{i_l}'\setminus\cup_{j=1}^{l-1}\mathcal{V}_{i_j}'|=d-k'-(l-1)$ (follows from $\mathcal{G}'$ structure) and $\{{\boldsymbol{\phi}}_i\}_{i\in \mathcal{K}''}$ comes from columns of a Cauchy matrix. It can be easily verified that nodes $\in$ $\mathcal{K}''$ will be providing {\em at least\em} $(d-k+1)(k-k')+{k-k' \choose 2}$ independent equations, in total. Note that each independent equation (over $\mathbb{F}_q$) corresponds to an independent evaluation of the underlying linearized polynomial. Thus, data-collector will be having access to at least $k'd-{k' \choose 2}+(d-k+1)(k-k')+{k-k' \choose 2}=kd-{k \choose 2}=B$ polynomial evaluations, at points independent over $\mathbb{F}_q$. Hence, data-collection is possible.

\end{proof}



\begin{figure}[t]
	\centering
	\captionsetup{justification=centering}
	\includegraphics[width=5in]{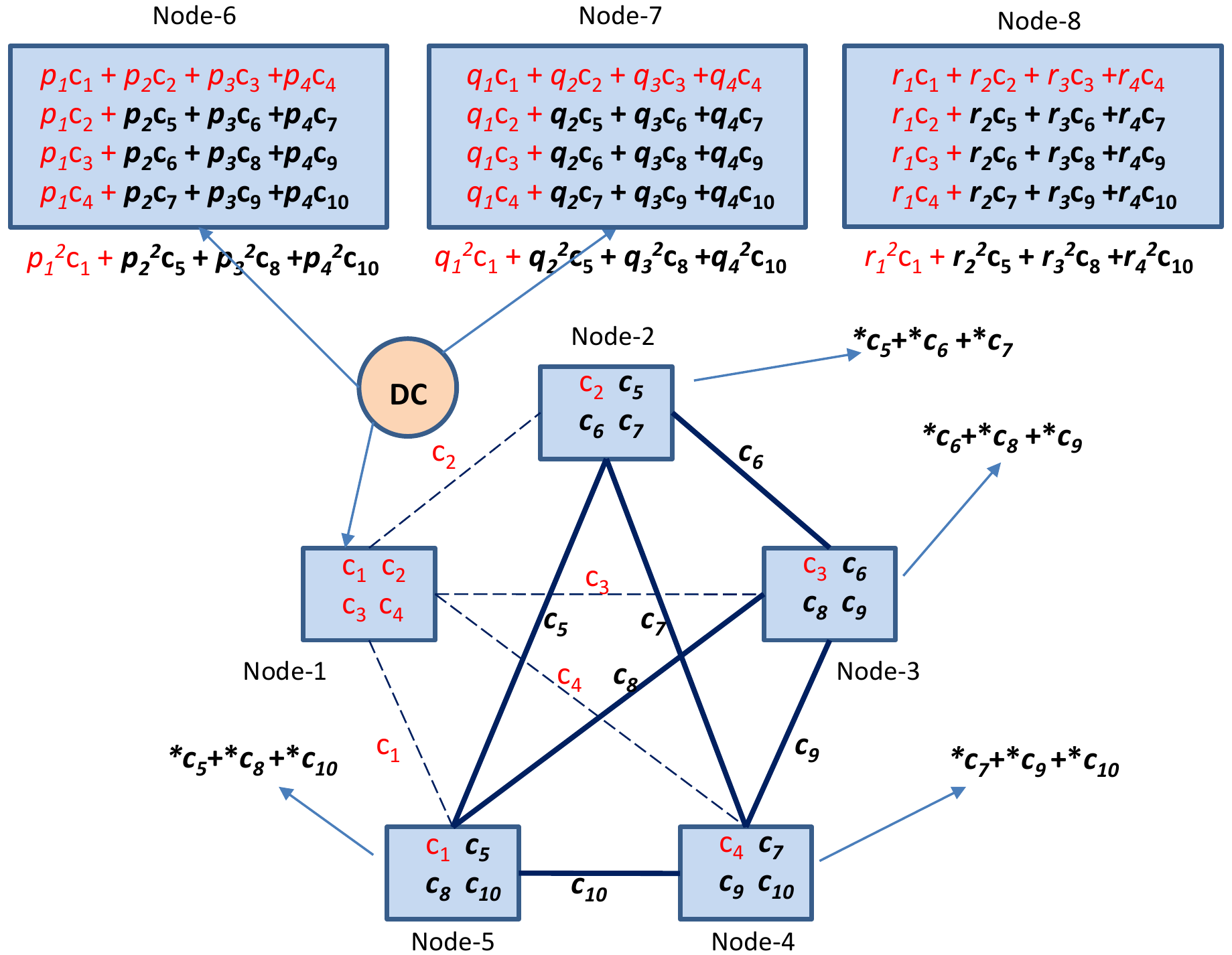}
	\caption{An example construction for $n=8, k=3, d=4$ Here, file-size is $9$. Data-collector is connected to nodes-$1$, $6$ and $7$. $k'=1$.}
	\label{fig:gen_example_n_eq_8}
\end{figure}

\subsection{Construction-A for $k=(n-3)$, $d=(n-2)$ over $\mathbb{F}_2$}{\label{k_eq_nm3_binary}
	Here, file-size, $B={(d+1) \choose 2}-1$. A parity symbol $p$ will be obtained by taking the XOR of  $\{m_i\}_{i=1}^{B}$. Define $c_i=m_i$, $1\leq i\leq B$ and $c_{(B+1)}=p$. These ${(d+1) \choose 2}$ symbols $\{c_i\}_{i=1}^{B+1}$ will be used as message symbols $\{m_i\}$ for the construction in \ref{deqkltnm2}, with ${\boldsymbol{\phi}}$ as in Remark-\ref{rem:binarykeqd}. The code thus formed will be MBR. Only DC property needs to be verified as node-repair property is inherited from the underlying construction.
	
	If all the $k$ nodes are read from $[n-1]$, DC follows from the MDS nature of the inner single parity check code. Consider the remaining case where $(k-1)=(n-4)$ nodes are read from
	$[n-1]$ along with node-$n$. After removing the known interference of symbols from $(n-4)$ nodes in $[n-1]$, there will be three sums on three unknowns (of which only two are independent) given by node-$n$; $(c_i+c_j)$, $(c_i+c_k)$ \& $(c_j+c_k)$. As $(c_i+c_j+c_k)$ is known from the single parity check nature of the inner code, $c_i$, $c_j$ and $c_k$  can be solved. An example is illustrated in Fig. \ref{fig:generalconstruction_n_5_k_2_d_3}.

	\begin{figure}[ht]
		\centering
		\captionsetup{justification=centering}
		\includegraphics[width=2in]{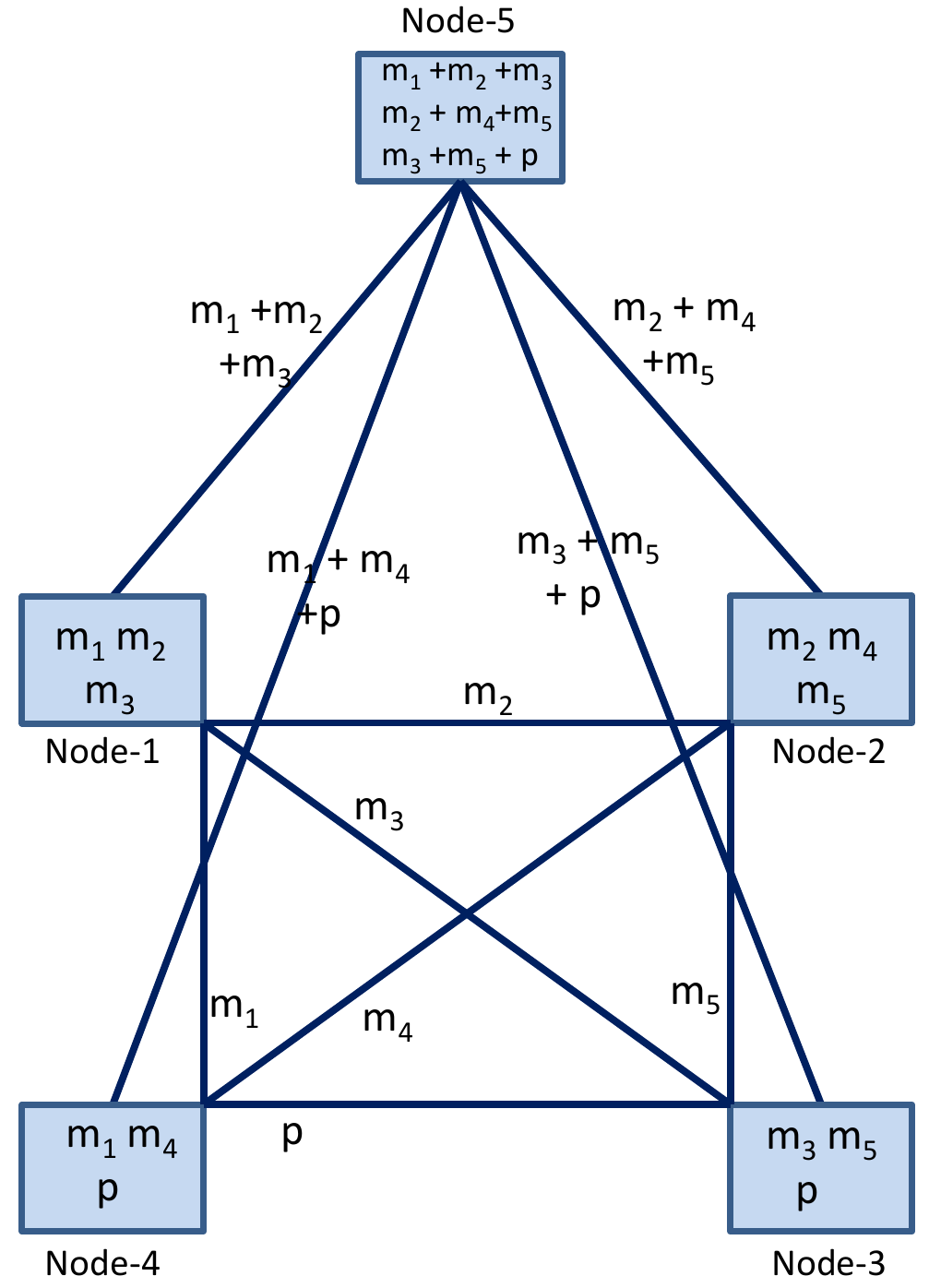}
		\caption{An example MBR construction over $\mathbb{F}_2$ for $(n=5, k=2, d=3)$.  File-size is $5$. All $m_i$'s have replication-$2$ and every solid edge coming to \{$1,2,3$\} carries HBT repair-data from the other end. }
		\label{fig:generalconstruction_n_5_k_2_d_3}
	\end{figure}

\subsection{Construction-B for $k\leq d\leq(n-1)$}\label{subsec:precodedPM}	

Shah et. al. \cite{NiharHBTnonexistence} gives an MBR construction whose $d$ nodes (these nodes contain a fraction of systematic message symbols) are repairable via HBT. Combining ideas from \cite{NiharHBTnonexistence} and \cite{ShaRasKumRam_rbt}, we first modify the conventional PM-MBR code discussed in \ref{PMdiscussion} to the following:
\begin{align}
{\bf{n}}_i &= {\boldsymbol{\chi}}(i)^t{\bf{M}}{\boldsymbol{\psi}}_i& 1\leq i\leq (d+1)\label{eq:PMnodecontentModified1}
\end{align}
where $\boldsymbol{\chi}(i)$ is a $(d \times d)$ invertible matrix whose columns are $\boldsymbol{\psi}_l$'s,  $1\leq l\leq (d+1),l\neq i$, in any order.

\begin{align}
{\bf{n}}_i &= {\boldsymbol{\chi}}(d+1)^t{\bf{M}}{\boldsymbol{\psi}}_i& (d+2)\leq i\leq n\label{eq:PMnodecontentModified2}
\end{align}

 These are invertible linear transformations of a node in the conventional PM-MBR code and hence, the code is linear MBR, with repair-data between two nodes same as that in \eqref{eq:PM_repairdata}. \eqref{eq:PMnodecontentModified1} implies that the first $(d+1)$ nodes have a structure identical to the RBT-MBR code in \cite{ShaRasKumRam_rbt}, with replication-$2$. \eqref{eq:PMnodecontentModified1} and \eqref{eq:PMnodecontentModified2} guarantee HBT property for repairs in $[d]$. We will be using an encoding matrix of the form \eqref{eq:PMsystematicMatrix}. $\boldsymbol{\psi}_i = \bf{e}_i$, $1\leq i\leq k$ where $\bf{e}_i$ denote $i^{th}$ column of the $(d\times d)$ identity matrix.

To get systematicity for the construction, we fill $\bf{M}$ with $B$ precoded symbols (structure as in \eqref{eq:PMmessagematrix}). First we consider a matrix $\bf{M}$$'$ filled with raw message symbols $\{m_i\}_{i=1}^{B}$ as in \eqref{eq:PMmessagematrix}. Let $\bf{M}$$_i$, $\bf{M}$$_{i}$$^{'}$ denote $i^{th}$ columns of $\bf{M}$ and $\bf{M}$$^{'}$ respectively. Both $\bf{M}$ and $\bf{M}$$^{'}$ are completely determined by their first $k$ columns (or rows). The precoder does the following set of $k$ operations to obtain $\bf{M}$ from $\bf{M}$$^{'}$.
\begin{align}
{\bf{M}}_i &= [{\boldsymbol{\Lambda}}(i)^t]^{-1}{\bf{M}}_i^{'} &  1 \leq i \leq k\label{eq:precodingPM}
\end{align}
where $\boldsymbol{\Lambda}(i)$ is a $(d\times d)$ invertible matrix whose $j^{th}$ column is given by:\[
	[\boldsymbol{\Lambda}(i)]_j = \left\{\begin{array}{lr}
	\boldsymbol{\psi}_{(d+1)} &  j=i\\
	\boldsymbol{\psi}_j & 1\leq j\leq d, j\neq i
	\end{array}\right.
	\]
 The equations in \eqref{eq:precodingPM} are indeed consistent, as they induce identity mapping for all $m_i$'s occurring more than once. The following example illustrates the precoder.

 	\begin{note}
 		\emph{(Update complexity)} Compared to PM-MBR, both the code families have lesser number of symbols modified for a single message symbol change, when $d=k$. 
    \end{note} 
 
 \begin{example}
 	Let the MBR parameters be $(n\geq4,k=2,d=3,\beta=1)$. $\therefore$ $B=5$ symbols.
    Let
 	\begin{equation*}
 	{\bf{M}'}= 
 	\begin{bmatrix}
 	m_1 & m_2 & m_3\\           
 	m_2 & m_4 & m_5\\
 	m_3 & m_5 & 0
 	\end{bmatrix},\ \boldsymbol{\psi}_1=[1\ 0\ 0]^t,\  \boldsymbol{\psi}_2=[0\ 1\ 0]^t
 	\end{equation*}
 	
 	$\therefore \boldsymbol{\Lambda}(1)=[\boldsymbol{\psi}_4\ \boldsymbol{\psi}_2\ \boldsymbol{\psi}_3]$, $\boldsymbol{\Lambda}(2)=[\boldsymbol{\psi}_1\ \boldsymbol{\psi}_4\ \boldsymbol{\psi}_3]$, ${\bf{M}}_1= {[\boldsymbol{\Lambda}(1)^t]}^{-1}{\bf{M}}_1^{'}$, ${\bf{M}}_2= {[\boldsymbol{\Lambda}(2)^t]}^{-1}{\bf{M}}_2^{'}$, $	\boldsymbol{\chi}_{1}=[\boldsymbol{\psi}_2\ \boldsymbol{\psi}_3\ \boldsymbol{\psi}_4]$, $\boldsymbol{\chi}_{2}=[\boldsymbol{\psi}_1\ \boldsymbol{\psi}_3\ \boldsymbol{\psi}_4]$. ${\bf{n}}_1$ = $\boldsymbol{\chi}_1^t\boldsymbol{M}\boldsymbol{\psi}_1$ =$\boldsymbol{\chi}_1^t\boldsymbol{M}_1$= $\boldsymbol{\chi}_1^t[\boldsymbol{\Lambda}(1)^t]^{-1}{\bf{M}}_1'$= $[m_2\  m_3\  m_1]^t$.
   Similarly, ${\bf{n}}_2$=$[m_2\  m_5\  m_4]^t$.
 	
 	\end{example}

\section{Construction of MBR codes with double replication for all symbols}\label{sec:MBRwithDoubleRep}

Consider a linear $(n,k,d,\alpha, \beta)$ MBR code over $\mathbb{F}$. Generator matrix, $\bf{G}=[G_1|G_2|\ldots|G_n]$, $\bf{G}$$  \in \mathbb{F}^{B\times n\alpha}$, $\bf{G_i}$$\in \mathbb{F}^{B\times \alpha}$. $\bf{n_i=G_i^tf}$, where $\bf{f}$ is the file-vector. We define the following subspaces, as in \cite{RasShaKumRam_allerton}.  The node subspace associated with node-$i$, $W_i=col(\bf{G_i})$, where $col(.)$ denotes column space. To repair a failed node-$i$, a helper node-$j$ will provide $\beta$ symbols. This will be viewed as passing a subspace of $W_j$ and termed a repair subspace of node-$j$. Node subspaces and repair subspaces have dimensions $\alpha$ and $\beta$ respectively \cite{RasShaKumRam_allerton}.
\begin{lem}\label{repairsubspaceIsNodeIntersection}
	The repair subspace from node-$i$ to node-$j$ is the intersection space of $W_i$ and $W_j$. 
\end{lem}
\begin{proof}
	Lemma-$3$ in  \cite{RasShaKumRam_allerton}.
\end{proof}
\begin{lem}\label{substitutionLemma}
	A linear MBR code can be transformed to another linear MBR code by substituting a node-$i$ content with repair-data for that node. The repair-data can come from an arbitrary set $\mathcal{D}$ of $d$ nodes and the substitution process is equivalent to re-encoding node-$i$ with a new basis set $\bf{G_i}'$ for $W_i$, instead of $\bf{G_i}$.
	\end{lem}
\begin{proof}
	Lemma-$3$ \& Theorem-$4$ in \cite{RasShaKumRam_allerton}.
\end{proof}

\begin{thm}\label{thm:MBRwithrep2existence}
	An $(n,k,d,\beta=1)$ MBR code with all symbols replicated twice exists iff there is a simple $d$-regular graph on $n$ vertices. Therefore, MBR codes with inherent double replication exists iff $nd$ is even.
\end{thm}
\begin{proof}
	If there is an MBR code with all symbols replicated twice, a corresponding $d$-regular graph can be constructed by putting an edge between every node sharing a symbol. From Lemma-\ref{repairsubspaceIsNodeIntersection}, it is clear that the graph  thus constructed will be simple, as $\beta=1$. Conversely, suppose there exists a simple $d$-regular graph $\mathcal{G}$ on $n$ vertices. Take an arbitrary $(n,k,d,\beta=1)$ MBR code and map it's nodes to $\mathcal{G}$. For each adjacent vertex-$j$ of vertex-$i$, the repair-data between node-$i$ and node-$j$ will be put in both node-$i$ and node-$j$, replacing a symbol of the starting MBR code in each of these two nodes. Repeating this for every node-$i$ will yield an MBR code (follows from Lemma-\ref{substitutionLemma}) with double replication. The second statement follows from Erdos-Gallai Theorem \cite{ErdosGallai}.
\end{proof}

\begin{cor}
	The constructions in \ref{deqkltnm2}, \ref{kltdltnm2}, \ref{subsec:precodedPM} having an RBT code in the first $(d+1)$ nodes, can be easily transformed to codes with double replication for all symbols, if $n\geq2(d+1)$ and $nd$ is even. These constructions have a compact description for the repair-data between nodes-$i,j$ $\in$ \{$(d+2),(d+3),\ldots,n$\}, which is of the form ${\bf{y}}_i^t{\bf{M}}{\bf{y}}_j$. However, this transformation will be at the expense of HBT repairability property of first $d$ nodes (from any other $d$ nodes).
\end{cor}
\begin{cor}
	In particular, if $(d+1)| n$, one can construct an MBR code, which will be a concatenation of $\frac{n}{(d+1)}$ RBT codes.
	\end{cor}

\begin{example}
	\emph{(Concatenated RBT)} Let the MBR parameters be $(n=6,k=2,d=2)$. Hence, file-size = $3$ symbols (over $\mathbb{F}_{2^2}$). Consider transforming the construction in \ref{deqkltnm2}. See Fig. \ref{fig:concatenatedRBT}.
	\begin{align}
	\boldsymbol{M} &= 
	\begin{bmatrix}
	m_1 & m_2 \\           
	m_2 & m_3
	\end{bmatrix} & 
	\boldsymbol{\phi} &= 
	\begin{bmatrix}
	1 & 1 & 1\\           
	1 & 2 & 3
	\end{bmatrix}\nonumber\nonumber
	\end{align}
\end{example}

$\therefore$ ${\boldsymbol{\phi}}_4=[1\ 1]^t$, ${\boldsymbol{\phi}}_5=[1\ 2]^t$ \& ${\boldsymbol{\phi}}_6=[1\ 3]^t$. ${\bf{n}}_1=[m_1\ m_2]^t$, ${\bf{n}}_2=[m_2\ m_3]^t$, ${\bf{n}}_3=[m_1\ m_3]^t$, ${\bf{n}}_4=[{\boldsymbol{\phi}}_5^tM{\boldsymbol{\phi}}_4\ {\boldsymbol{\phi}}_6^tM{\boldsymbol{\phi}}_4]^t=[m_1+3m_2+2m_3\ \ m_1+2m_2+3m_3]^t$, ${\bf{n}}_5=[{\boldsymbol{\phi}}_4^tM{\boldsymbol{\phi}}_5\ {\boldsymbol{\phi}}_6^tM{\boldsymbol{\phi}}_5]^t=[m_1+3m_2+2m_3\ \ m_1+m_2+m_3]^t$ and ${\bf{n}}_6=[{\boldsymbol{\phi}}_4^tM{\boldsymbol{\phi}}_6\ {\boldsymbol{\phi}}_5^tM{\boldsymbol{\phi}}_6]^t=[m_1+2m_2+3m_3\ \ m_1+m_2+m_3]^t$. \{${\bf{n}}_1$, ${\bf{n}}_2$, ${\bf{n}}_3$\}, \{${\bf{n}}_4$, ${\bf{n}}_5$, ${\bf{n}}_6$\} form two RBT codes, where first one has systematic data.

\begin{figure}[t]
	\centering
	\captionsetup{justification=centering}
	\includegraphics[width=4in]{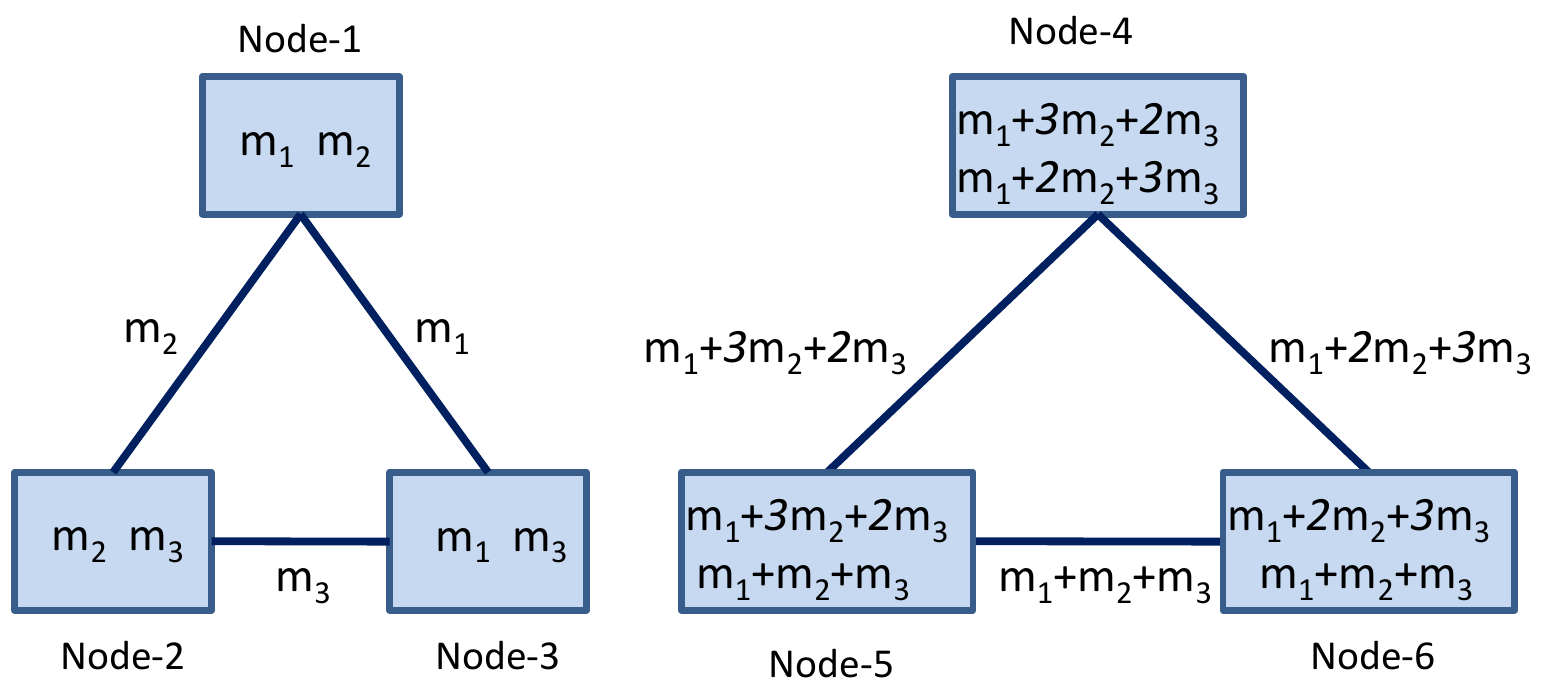}
	\caption{An example construction for $n=6, k=2, d=2$.}
	\label{fig:concatenatedRBT}
\end{figure}

%
%
%
%
%

\begin{thm}
	If $nd$ is odd, it is possible to construct an $(n,k,d,\beta=1)$ MBR code, where all but one symbols have replication $2$.
\end{thm}
\begin{proof}
	Let $d'=(d-1)$. Construct a $d'$-regular graph, $\mathcal{G}_1$ on $n$ vertices. Now, remove vertex-$i$ and construct a $1$-regular graph $\mathcal{G}_2$ on the remaining $(n-1)$ vertices. Form a graph union $\mathcal{G}=\mathcal{G}_1\cup\mathcal{G}_2$ and map the vertices to nodes of an  arbitrary $(n,k,d,\beta=1)$ MBR code. Using the technique in proof of Theorem-\ref{thm:MBRwithrep2existence}, all except one symbol in the transformed code will have replication-$2$. For node-$i$ with degree $d'$, repair-data coming from a non-adjacent node (w.r.t $\mathcal{G}$) will be taken as the $d^{th}$ symbol. 
\end{proof}

\bibliographystyle{IEEEtran}
\bibliography{MBRCodesWithReplication}

\end{document}